\documentclass[a4paper,11pt,leqno]{article}
\usepackage{amsmath}
\usepackage{amsfonts}
\usepackage{eucal}
\usepackage{amsthm}
\numberwithin{equation}{section}

\newcommand{\bigpare}[1]{\bigl(#1\bigr)}
\newcommand{\biggpare}[1]{\biggl(#1\biggr)}
\newcommand{\Bigpare}[1]{\Bigl(#1\Bigr)}

\newcommand{\bigbrac}[1]{\bigl[#1\bigr]}

\newcommand{\bigset}[2]{\bigl\{#1\bigm|#2\bigr\}}

\newcommand{\norm}[1]{\| #1 \|}
\newcommand{\bignorm}[1]{\bigl\| #1 \bigr\|}

\newcommand{\biggnorm}[1]{\biggl\| #1 \biggr\|}

\newcommand{\bigabs}[1]{\bigl| #1 \bigr|}

\newcommand{\jap}[1]{\langle #1 \rangle}

\def\a{\alpha}
\def\b{\beta}
\def\c{\gamma}
\def\d{\delta}
\def\e{\varepsilon}
\def\f{\varphi}
\def\g{\psi}

\def\l{\lambda}
\def\m{\mu}
\def\n{\nu}

\def\s{\sigma}

\def\x{\xi}

\def\z{\zeta}
\def\th{\theta}
\newcommand{\F}{\Phi}
\newcommand{\G}{\Psi}

\def\re{\mathbb{R}}
\def\co{\mathbb{C}}
\def\ze{\mathbb{Z}}

\def\pa{\partial}

\renewcommand{\Re}{\text{{\rm Re}\;}}
\renewcommand{\Im}{\text{{\rm Im}\;}}

\newcommand{\supp}{\text{{\rm supp}\;}}

\DeclareMathOperator*{\slim}{s-lim}

\newcommand{\Ran}{\text{\rm Ran\;}}

\newtheorem{thm}{Theorem}
\newtheorem{lem}[thm]{Lemma}
\newtheorem{prop}[thm]{Proposition}
\newtheorem{cor}[thm]{Corollary}

\theoremstyle{definition}

\newtheorem{ass}{Assumption}

\theoremstyle{remark}

\title{Time-dependent scattering theory for Schr\"odinger operators %
on scattering manifolds}
\author{
Kenichi I{\sc to}%
\footnote{Graduate School of Pure and Applied Sciences, University of Tsukuba,
1-1-1 Tennodai, Tsukuba Ibaraki, 305-8571 Japan. 
E-mail: \texttt{ito\_ken@math.tsukuba.ac.jp} }
\ and Shu N{\sc akamura}%
\footnote{Graduate School of Mathematical Sciences, 
University of Tokyo, 3-8-1 Komaba, Meguro Tokyo, 
153-8914 Japan. 
E-mail: {\tt shu@ms.u-tokyo.ac.jp}.  \newline
Partially supported by JSPS Grant Kiban (B) 17340033.} }
\begin{document}
\maketitle

\begin{abstract}
We construct a time-dependent scattering theory for Schr\"odinger operators 
on a manifold $M$  with asymptotically conic structure. We use
the  two-space scattering theory formalism, and a reference operator  
on a space of the form $\re\times \pa M$, where $\pa M$ is the boundary of $M$ 
at infinity. We prove the existence and 
the completeness of the wave operators, and show that our scattering 
matrix is equivalent to the absolute scattering matrix, which is defined in terms of the 
asymptotic expansion of generalized eigenfunctions. Our method is functional 
analytic, and we use no microlocal analysis in this paper. 
\end{abstract}

%%%%%%%%%%%%%%%%%%%%%%%%%%%%%%%%%%%%%%%%%
%%%%%%%%%%%%  Section 1  %%%%%%%%%%%%%%%%
%%%%%%%%%%%%%%%%%%%%%%%%%%%%%%%%%%%%%%%%%

\section{Introduction}

In this paper, we consider the scattering theory for Schr\"odinger operators on an 
aymptotically conic manifold, which we call a {\it scattering manifold} following 
Melrose~\cite{Mel1, Mel2}. We propose a construction of a time-dependent scattering theory 
for this model using a simple reference system and the two-space scattering framework of 
Kato~\cite{Ka}. Here we present the formulation and results in a relatively simple 
situation to exhibit the main idea of the method. Namely, we are not trying to 
formulate and prove the most general theorems. There are many directions 
to generalize and/or refine our results, and we hope to study several 
possibilities in forthcoming research projects. 

Let $M$ be an $n$-dimensional smooth non-compact manifold such that
$M$ is decomposed to $M_c\cup M_\infty$, where $M_c$ is relatively compact, 
and $M_\infty$ is diffeomorphic to $\re_+\times \pa M$ with a compact manifold 
$\pa M$. We fix an identification map:
\[
\iota \ :\ M_\infty \ \longrightarrow \ \re_+\times \pa M \in (r,\th). 
\]
We suppose $M_c\cap M_\infty\subset (0,1)\times \pa M$ under this identification. 
We also suppose $\pa M$ is equipped with a measure $H(\th)d\th$ where $H(\th)$ is a 
smooth positive density, and a positive (2,0)-tensor $h=(h^{jk}(\th))$. We set
\[
\mathcal{H}_b = L^2(\pa M, H d\th), \quad 
Q=-\frac12 \sum_{j,k} H(\th)^{-1}\frac{\pa}{\pa\th_j} H(\th) h^{jk}(\th) \frac{\pa}{\pa \th_k}, 
\]
where $(\th_j)_{j=1}^{n-1}$ denotes a local coordinate system in $\pa M$. 
We note $Q$ is an essentially self-adjoint elliptic operator on $\mathcal{H}_b$. As usual, 
we denote the unique self-adjoint extension by the same symbol $Q$. 

Let $\{\f_\a: U_\a \to \re^{n-1}\}$, $U_\a\subset \pa M$, be a local coordinate system of $\pa M$. 
We set $\{\tilde\f_\a: \re_+\times U_\a \to \re\times\re^{n-1}\}$ to be a local coordinate
system of $M_\infty \cong \re_+\times \pa M$, and we denote $(r,\th)\in \re\times \re^{n-1}$ to 
represent a point in $M_\infty$. We suppose $G(x)$ is a smooth positive density on $M$ such that 
\[
G(x)dx= r^{n-1} H(\th) dr\,d\th \quad \text{on } \tilde M_\infty =(1,\infty)\times \pa M
\subset M_\infty, 
\]
and we set
\[
\mathcal{H}= L^2(M, G(x) dx).
\]
Let $P$ be a formally self-adjoint second order elliptic operator on $\mathcal{H}$ of the form: 
\[
P= -\frac12 G^{-1} (\pa_r, \pa_\th/r) G \begin{pmatrix} a_1 & a_2 \\ {}^t a_2 & a_3 \end{pmatrix}
\begin{pmatrix} \pa_r \\ \pa_\th/r \end{pmatrix} +V
\quad \text{on } \tilde M_\infty, 
\]
where $a_1$, $a_2$, $a_3$ are real-valued smooth tensors, and $V$ is a real-valued smooth function. 

\begin{ass}
For any $\ell\in \ze_+$, $\a\in\ze_+^{n-1}$, there is $C_{\ell \a}>0$ such that 
\begin{align*}
&\bigabs{\pa_r^\ell\pa_\th^\a (a_1(r,\th)-1)} \leq C_{\ell\a} r^{-\m_1-\ell}, \\
&\bigabs{\pa_r^\ell\pa_\th^\a a_2(r,\th)} \leq C_{\ell\a} r^{-\m_2-\ell}, \\
&\bigabs{\pa_r^\ell\pa_\th^\a (a_3(r,\th)-h(\th))} \leq C_{\ell\a} r^{-\m_3-\ell},\\
&\bigabs{\pa_r^\ell\pa_\th^\a V(r,\th)} \leq C_{\ell\a} r^{-\m_4-\ell}
\end{align*}
on $\tilde M_\infty$, where $\m_1,\m_4>1$ and $\m_2,\m_3>0$. 
Note that we use the coordinate system in $M_\infty$ described above. 
\end{ass}

Roughly speaking, Assumption~A implies the principal term of $P$ is asymptotically close to 
$-\frac12 (\pa_r^2+\frac1{r^2} Q)$ as $r\to\infty$. 
The condition on the scalar potential $V$ can be considerably relaxed, and the precise 
condition is discussed in the following sections. 
We suppose Assumption~A throughout this section. 

We will construct a time-dependent scattering theory for $P$ on $\mathcal{H}$. 
This model is a natural generalization of the Laplacian on a manifold with scattering metric
(cf. \cite{Mel1}), and we discuss the geometric construction and the relationship  
with the scattering metric in Appendix~A. 

We first observe several basic spectral properties of $P$, which are quite analogous to 
those for Schr\"odinger operators on $\re^n$. We denote the essential spectrum, 
the discrete spectrum, the pure point spectrum, the absolutely continuous spectrum, 
and the singular continuous spectrum by $\s_{ess}(\cdot)$, $\s_{d}(\cdot)$, 
$\s_{pp}(\cdot)$, $\s_{ac}(\cdot)$ and $\s_{sc}(\cdot)$, respectively. 

\begin{thm}
(1) $P$ is essentially self-adjoint on $C_0^\infty(M)$. We denote the unique self-adjoint 
extension by the same symbol $P$. \newline
(2) $\s_{ess}(P)=[0,\infty)$, and $\s_{disc}(P) \subset (-\infty,0]$ is bounded and 
accumulate at most at $\{0\}$. \newline
(3) $\s_{pp}(P)$ is discrete except for $\{0\}$; $\s_{ac}(P)=[0,\infty)$; $\s_{sc}(P)=\emptyset$.
\end{thm}

We prove Theorem~1 under a weaker assumption in Section~2. 

In order to construct a time-dependent scattering theory, we define a reference system as follows: 
\begin{align*}
& M_f =\re\times \pa M, \quad \mathcal{H}_f = L^2 (M_f, H(\th) dr\,d\th), \\
& P_f = -\frac12\frac{\pa^2}{\pa r^2} \quad\text{on } M_f.
\end{align*}

It is easy to show that $P_f$ is essentially self-adjoint on $C_0^\infty(M_f)$, and we denote 
the unique self-adjoint extension by the same symbol as before. Let $j(r)\in C^\infty(\re)$ 
such that $j(r)=1$ if $r\geq 1$ and $j(r)=0$ if $r\leq 1/2$. We define $J: \mathcal{H}_f\to
\mathcal{H}$ by 
\[
(J\f)(r,\th) = r^{-(n-1)/2}\, j(r) \, \f(r,\th)
\quad \text{if }(r,\th)\in M_\infty, 
\]
and $J\f(x)=0$ if $x\notin M_\infty$, where $\f\in \mathcal{H}_f$. 
We note if $\supp\f\subset (1,\infty)\times\pa M$, then $\norm{J\f}_\mathcal{H} 
=\norm{\f}_{\mathcal{H}_f}$. We consider the two-space wave operators: 
\[
W_\pm = W_\pm(P,P_f;J) = \slim_{t\to\pm \infty} e^{itP} J e^{-itP_f}.
\]
Note we do not need the projection to the absolutely continuous subspace of $P_f$ 
since $P_f$ is absolutely continuous. We denote the Fourier transform with respect to 
$r$-variable by $\mathcal{F}$:
\[
\mathcal{F} \f(r,\rho) = (2\pi)^{-1/2}\int_{-\infty}^\infty e^{-i r\rho} \f(r,\th) dr, 
\quad \text{for }\f\in C_0^\infty(M_f). 
\]
We decompose the reference Hilbert space $\mathcal{H}_f$ as 
$\mathcal{H}_f = \mathcal{H}_f^+\oplus \mathcal{H}_f^-$, where $\mathcal{H}_f^\pm$ 
are defined by 
\[
\mathcal{H}_f^\pm =\bigset{\f\in\mathcal{H}_f}{\supp(\mathcal{F}\f) 
\subset [0,\infty)\times \pa M}.
\]
Then we have the following existence and the completeness of the wave operators: 

\begin{thm}
(1) The wave operators $W_\pm=W_\pm(P,P_f;J)$ exist. 
Moreover, $W_\pm \mathcal{H}_f^\mp =0$ and $W_\pm\lceil_{\mathcal{H}_f^\pm}$ 
are isometries from $\mathcal{H}_f^\pm$ into $\mathcal{H}$. \newline
(2) $W_\pm$ are complete, i.e., $\Ran W_\pm= \mathcal{H}_{ac}(P)$, where 
$\mathcal{H}_{ac}(\cdot)$ denotes the absolutely continuous subspace. 
\end{thm}

We prove Theorem~2 in Section~3. 

Now we can define the scattering operator by 
\[
S= W_+^* W_-\ :\ \mathcal{H}_f^- \to \mathcal{H}_f^+, 
\]
and $S$ is unitary. By the intertwining property: $PW_\pm =W_\pm P_f$ and
the fact that $\mathcal{F}$ is a spectral representation of $P_f$, we learn that there 
exist unitary operators $S(\l)$ on $\mathcal{H}_b$ for a.e.\ $\l>0$ such that 
\[
(\mathcal{F} S\f)(\rho,\th) = (S(\rho^2/2)(\mathcal{F}\f)(-\rho,\cdot))(\th) 
\quad \text{for }\rho>0, \ \f\in \mathcal{H}_f^-. 
\]
$S(\l)$ is called the {\it scattering matrix}, and we will see that it is equivalent to the 
so-called {\it absolute scattering matrix}, which is defined in terms of the asymptotic expansion 
of generalized eigenfunctions as $r\to\infty$. We discuss these properties of the 
scattering matrix in Section~4. 

\medskip
Scattering theory for Schr\"odinger operators has long history, and huge amount of 
work have been done, both in abstract setting and for concrete models (see, e.g., 
\cite{RS} Vol.~III or \cite{Yaf} and refererences therein). 
Most models of scattering theory are concerned with operators on a Euclidean 
space, or on the product space of a compact space and a Euclidean space, 
though there have been several works on asymptotic expansion of solutions 
on certain manifolds (e.g., Wang \cite{Wa}).

 In early 1990's, 
Melrose introduced a new framework of scattering theory on a class of Riemannian 
manifolds with metrics called {\em scattering metrics}\/ (\cite{Mel1}, see also \cite{Mel2, 
Va, HaVa1, HaVa2, Jo, ChJo}). 
He and the other  authors have studied Laplace-Beltrami operators on such manifolds 
using a pseudodifferential operator calculus, which is sometimes called the 
{\em scattering calculus}. They have also studied the absolute scattering 
matrix, which is defined through the asymptotic expansion of generalized eigenfunctions. 
However, as far as the authors are aware of, there have been no works on 
time-dependent scattering theory on manifolds with scattering metrics, 
most likely because there is no obvious choice of the {\em free}\/ Schr\"odinger 
operators. 

When the authors worked on the microlocal singularities of solutions to Schr\"odinger 
operators on scattering manifolds (\cite{ItNa1}, see also Hassell-Wunsch \cite{HaWu}), 
we constructed a classical scattering theory on scattering manifolds, and also 
a reference quantum system to describe the propagation of singularities. 
It turns out that this idea can be used to construct a time-dependent 
quantum scattering theory, which is the subject of this paper. As we discuss in 
Section~4, the (usual) scattering matrix in our framework is equivalent to the 
absolute scattering matrix, and we are currently working on an alternative proof of the 
Melrose-Zworski theorem on the microlocal properties of the scattering 
matrix (\cite{MZ}). 

Our method is quite simple and flexible, and hence we expect it can be 
applied to various models, including hyperbolic manifolds and non-compact 
Riemannian manifolds with polynomial growth at infinity. 
We may also consider scattering theory not only for functions, but also
for differential forms. About another direction of generalization, 
it should be possible to consider more general perturbations, 
including long-range potentials and magnetic fields. 
We expect many methods in the traditional scattering theory can be 
applied to operators on scattering manifolds (or more general 
non-compact manifolds), and there remain various problems 
to be addressed. 

\medskip 
We use the following notation throughout the paper: 
The definition domain of an operator $A$ is denoted by $\mathcal{D}(A)$. 
$\mathcal{L}(X,Y)$ denotes the Banach space of the bounded operators 
from a norm space $X$ to a Banach space $Y$, and 
$\mathcal{L}(X)=\mathcal{L}(X,X)$. For a self-adjoint operator $A$, 
$E_A(\cdot)$ denotes the spectral measure of $A$. 
For $x\in\re^m$, $m>0$, we write $\jap{x}=\sqrt{1+|x|^2}$. For $x\in M$, 
we also write
\[
\jap{x}=\jap{r}= \begin{cases} 1+r j(r) \quad&\text{for } x\in M_\infty, \\
1 &\text{for } x\in M_c. \end{cases}
\]
$\chi_\Omega(x)$ denotes the indicator function of $\Omega$, i.e., 
$\chi_\Omega(x)=1$ for $x\in\Omega$, and $=0$ for $x\notin \Omega$. 
The positive and negative real axes are denoted by $\re_+=(0,\infty)$ and
$\re_-=(-\infty,0)$, respectively. The set of the nonnegative integers is 
denoted by $\ze_+$, and $\ze_+^n$ denotes the set of the multi-indeces.

%%%%%%%%%%%%%%%%%%%%%%%%%%%%%%%%%%%%%%%%%
%%%%%%%%%%%%  Section 2  %%%%%%%%%%%%%%%%
%%%%%%%%%%%%%%%%%%%%%%%%%%%%%%%%%%%%%%%%%

\section{Self-adjointness and spectral properties}

The spectral properties of Schr\"odinger operators on scattering manifolds 
as described in Theorem~1 seem to be more or less known, 
but here we give proof for the sake of completeness. 

Let $P_0$ be the second order elliptic operator on $\mathcal{H}$  satisfying 
Assumption~A. We define 
\[
H^2(M) = \mathcal{D}(P_0) = \bigset{\f\in\mathcal{H}}{P_0\f\in\mathcal{H}}
\]
with its graph norm, where $P_0$ acts on elements of $\mathcal{H}$ in the distribution sense. 
This definition is independent of the choice of $P_0$ up to equivalence, as long as the coefficients of 
$P_0$ satisfy Assumption~A. Typically we may choose $a_1\equiv 1$, $a_2\equiv 0$ and 
$a_3\equiv h$ on $\tilde M_\infty$, and extend the coefficients smoothly so that 
$P_0$ is globally elliptic. 

\begin{prop}
$P_0$ is self-adjoint with $\mathcal{D}(P_0)=H^2(M)$, and $C_0^\infty(M)$ is a core for $P_0$. 
In particular, $P_0$ is essentially self-adjoint on $C_0^\infty$. 
\end{prop}

\begin{proof}
We first show $C_0^\infty(M)$ is a core. Let $u\in H^2(M)$. Then by the 
elliptic regularity theorem, we learn $u\in H^2_{loc}(M)$ in the sense of the usual Sobolev spaces. 
Then by a simple interpolation, we have $\pa_r u\in L^2(M_\infty)$. We now set
\[
u_R =(1-j(r/R))u, 
\]
where $j(r)$ is the smooth cut-off function in the previous section (we again identify $M_\infty$ with 
$\re_+\times \pa M$). Then it is easy to show
\[
\norm{u-u_R}_\mathcal{H} \to 0, \quad \norm{P_0(u-u_R)}_\mathcal{H} \to 0
\quad \text{as }R\to\infty.
\]
Thus for any $\e>0$, $\norm{u-u_R}_{H^2(M)}<\e/2$ provided $R$ is sufficiently large. 
By the well-known fact that $C_0^\infty(K)$ is dense in $H^2(K)$ for $K\subset\subset M$, 
we learn there is $u'_\e\in C_0^\infty(M)$ such that $\norm{u_R-u'_\e}_{H^2(M)}<\e/2$, 
and hence $\norm{u-u'_\e}_{H^2(M)}<\e$. 

By the definition, we easily see $\mathcal{D}(P_0^*)\subset H^2(M)= \mathcal{D}(P_0)$, and 
hence it remains only to show $P_0$ is symmetric. If $u,v\in C_0^\infty(M)$, then 
$(P_0u,v)=(u,P_0v)$ follows from the Stokes theorem (and the the formal self-adjointness of $P_0$). 
Then by the density argument and the result above, 
we conclude the equality holds for any $u,v\in \mathcal{D}(P_0)$. 
\end{proof}

We now apply the standard perturbation theory for self-adjoint operators, 
i.e., the Kato-Rellich theorem and the Weyl theorem on the essential spectrum
(see, e.g., \cite{RS} Section~X.2 and Section~XIII.4, respectively). 

\begin{cor}
(1) Suppose $V$ is a symmetric operator on $\mathcal{H}$, and $V$ is $P_0$-bounded with 
the relative bound less than one, i.e., there exist $\b<1$ and $C>0$ such that 
\[
\norm{Vu}_\mathcal{H} \leq \b \norm{P_0 u}_{\mathcal{H}} + C\norm{u}_\mathcal{H}
\quad \text{for }u\in H^2(M). 
\]
Then $P=P_0+V$ is self-adjoint with $\mathcal{D}(P)=H^2(M)$, and $C_0^\infty(M)$ is a core. 
Moreover, $P$ is bounded from below. In particular, if $V=V_1+V_2$ where $V_1$ is compactly 
supported, $V_1\in L^p(M,Gdx)$ with $p>n/2$, $p\geq 2$, and $V_2\in L^\infty(M,Gdx)$, 
then $P$ is self-adjoint. \newline
(2) Suppose, in addition to the conditions above, $V$ is $P_0$-compact, i.e., if $V$ is a 
compact operator from $H^2(M)$ to $\mathcal{H}$, then $\s_{ess}(P)=\s_{ess}(P_0)
=\re_+$. In particular, the condition holds if $V=V_1+V_2$ as in (1) and 
$V_2(x)\to 0$ as $r\to\infty$. 
\end{cor}

We can easily show  $\s_{ess}(P_0)=[0,\infty)$ since $P_0\geq 0$ and 
it is straightforward to construct a Weyl sequence to show $E\in \s_{ess}(P_0)$ for $E>0$. 

The sufficient condition on $V$ given in Corollary 4 is not optimal. It would be appropriate 
to use a generalization of the Stummel class (see, e.g.,  \cite{CFKS} Section~1.2), 
though we do not investigate it in this paper. 

We next study the spectral properties of $P$ using the Mourre theory \cite{Mo}. 
Let $R>0$, which will be fixed later. Let $X$ be a vector field on $M$ such that 
\[
X=j(\tfrac{r}{R})\, r\, \frac{\pa}{\pa r} \quad \text{on } M_\infty, 
\]
and $X=0$ on $M_c$. Let $\{\exp[tX]\,|\,t\in\re\}$ be the flow on $M$ generated by $X$. 
The flow induces a one-parameter unitary group defined by 
\[
U(t)\f(x) = \F(t,x) \f(\exp[-tX](x)) \quad \text{for }\f\in \mathcal{H},
\]
where $\F(t,x)$ is a weight function to make $U(t)$ unitary. Let $A$ be the generator 
of $U(t)$, i.e., $U(t) =e^{itA}$. By straightforward computation, we learn 
\begin{align*}
A &= \frac{1}{2i} G(x)^{-1/2} \biggpare{j(\tfrac{r}{R}) \, r \, \frac{\pa}{\pa r} 
+ \frac{\pa}{\pa r} j(\tfrac{r}{R})\,r} G(x)^{1/2} \\
&= \frac{1}{2i} \biggpare{j(\tfrac{r}{R})\, r\, \frac{\pa}{\pa r}
+ \frac{\pa}{\pa r} j(\tfrac{r}{R})\, r + \frac{n-1}{2} j(\tfrac{r}{R})}
\end{align*}
on $M_\infty$, and $A=0$ on $M_c$. 

\begin{ass} $V=V_1+V_2$ such that $V_1$ is compactly supported, $V_1\in L^p(M,Gdx)$ 
with $p>n/2$, $p\geq 2$, and $V_2\in C^2(M)$ such that 
\[
\bigabs{\pa_r^\ell V(r,\th)} \leq C r^{-\ell-\n}, 
\quad (r,\th)\in \tilde M_\infty, j=0,1,2, 
\]
with $\n>0$, $C>0$. 
\end{ass}

We now fix $R>0$ so that $j(\frac{r}{R})V_1(r,\th)\equiv 0$. Then $P$ and $A$ satisfy the conditions 
to apply the Mourre theory:

\begin{lem}
Let $P=P_0+V$ where $P_0$ satisfies Assumption~A and $V$ satisfies Assumption~B. Then: 
\newline
(a) $C_0^\infty(M)\subset \mathcal{D}(A)\cap \mathcal{D}(P)$ and $C_0^\infty(M)$ is a core 
for $P$. \newline
(b) $U(t)=e^{itA}$ leaves $\mathcal{D}(P)$ invariant, and 
\[
\sup_{|t|<1} \norm{P e^{itA}u}_\mathcal{H}<\infty \quad \text{for }u\in\mathcal{D}(P).
\]
(c) $i[P,A]$ is extended to a semi-bounded self-adjoint operator, and it is $P$-bounded. \newline
(d) $[[P,A],A]$ is bounded from $H^2(M)$ to $H^2(M)^*$. \newline
(e) Let $I\subset \re_+$. Then there exist $\b>0$ and a compact operator $K$ on $\mathcal{H}$ 
such that 
\[
\chi_I(P) i[P,A]\chi_I(P) \geq \b \chi_I(P) + \chi_I(P) K\chi_I(P), 
\]
where $\chi_I$ is the indicator function of $I$. 
\end{lem}

Then we obtain the following spectral properties of $P$: 

\begin{thm}
Let $P=P_0+V$ where $P_0$ satisfies Assumption~A and $V$ satisfies Assumption~B. 
Then \newline
(1) The point spectrum of $P$ is discrete in $\re_+$. \newline
(2) Let $I\subset \re_+\setminus \s_{pp}(P)$. Then 
\[
\sup_{\Re \! z\in I, \Im \! z\neq 0} \norm{\jap{A}^{-1} (P-z)^{-1} \jap{A}^{-1}} <\infty.
\]
In particular, $\s_{sc}(P)=\emptyset$. \newline
(3) Let $I$ as above, and let $s>1/2$. Then 
\[
\jap{r}^{-s} (P-E\pm i0)^{-1}\jap{r}^{-s} 
=\lim_{\e\to+0}\jap{r}^{-s} (P-E\pm i\e)^{-1}\jap{r}^{-s}
\]
converges uniformly in $E\in I$. 
\end{thm}

The claims  (1) and (2) follows by the standard Mourre theory (see, e.g., \cite{Mo}, 
\cite{CFKS} Chapter~4), and the claim (3) follows by its refinement by 
Perry-Sigal-Simon (\cite{PSS}, see also \cite{ABG}). We note that by using methods of 
Amrein-Boutet de Monvel-Georgescu \cite{ABG} or Tamura \cite{Ta}, we can weaken 
the condition on $V$ considerably.

%%%%%%%%%%%%%%%%%%%%%%%%%%%%%%%%%%%%%%%%%
%%%%%%%%%%%%  Section 3  %%%%%%%%%%%%%%%%
%%%%%%%%%%%%%%%%%%%%%%%%%%%%%%%%%%%%%%%%%

\section{Wave operators and the asymptotic completeness}

In this section, we prove Theorem~2, i.e., the existence and the asymptotic completeness 
of wave operators. Throughout this and the next section, we suppose $P=P_0+V$, where 
$P_0$ satisfies Assumption~A and $V$ satisfies Assumption~B with $\n>1$. 
The scalar potential term is denoted by $V$ collectively.
We may suppose $V_1$, the singular part of the potential $V$, 
is supported in $M_c$ without loss of generality. 

Following the standard procedure, we set
\[
T= PJ-JP_f\ :\ \mathcal{D}(P_f)\cap \mathcal{D}(PJ) \to \mathcal{H}.
\]
$T$ plays the role of the perturbation term in the two-space scattering theory. 
By simple computations, we have 
\[
T\f = -\frac12 G^{-1} (\pa_r,\pa_\th/r)G \begin{pmatrix} a_1-1 & a_2 \\ {}^ta_2 & a_3
\end{pmatrix} \begin{pmatrix}\pa_r \\ \pa_\th/r \end{pmatrix} J\f +VJ\f
\]
for $\f\in\mathcal{D}(P_f)\cap \mathcal{D}(Q)\subset 
\mathcal{D}(P_f)\cap \mathcal{D}(PJ)$. For $\f\in\mathcal{D}(P_f)\cap\mathcal{D}(Q)$, 
we set
\[
|||\f ||| =\norm{\f}_{\mathcal{H}_f}+\norm{P_f\f}_{\mathcal{H}_f}+\norm{Q\f}_{\mathcal{H}_f}.
\]

\begin{lem}
Let $\m\leq \min(\m_1,\m_2+1,\m_3+2,\m_4,\n)>1$. Then there is $C>0$ such that 
\[
\norm{T\jap{r}^\m\f}_{\mathcal{H}} \leq C|||\f|||
\]
for $\f\in \mathcal{D}(P_f)\cap\mathcal{D}(Q)$. 
\end{lem}

The proof is an easy computation using the assumptions. 
Now the existence of the wave operators is straightforward consequence of the 
Cook-Kuroda method: 

\begin{thm}
The wave operators 
\[
W_\pm = W_\pm(P,P_f;J) = \slim_{t\to\pm\infty} e^{itP} J e^{-itP_f}
\]
exist. Moreover, $W_\pm\mathcal{H}^\mp_f =0$, and $W_\pm$ are isometry 
from $\mathcal{H}_f^\pm$ into $\mathcal{H}$. 
\end{thm}

\begin{proof}
For $\f\in\mathcal{D}(P_f)\cap\mathcal{D}(Q)$, we have 
\[
\frac{d}{dt} e^{itP} J e^{-itP_f}\f =i e^{itP} T e^{-itP_f}\f,
\]
and hence 
\begin{equation}
e^{-itP} J e^{-itP_f} \f = J\f +i \int_0^t e^{isP} T e^{-isP_f}\f ds.
\end{equation}
At first we suppose $\mathcal{F}\f\in C_0^\infty(\re\times \pa M)$ and 
$\supp\!(\mathcal{F}\f)\subset (-\infty,-\d)\times \pa M$ with some $\d>0$. 
Such $\f$'s are dense in $\mathcal{H}_f^-$. 
Then by the non-stationary phase method (see, e.g., \cite{RS} Section~XI.3, 
\cite{Ho} Section~7.7), 
or by the integration by parts, we learn 
\[
\bignorm{\chi_{(-\d t/2,\infty)}(r)e^{-itP_f}\f} \leq C_N\jap{t}^{-N}, \quad t>0,
\]
for any $N$, and in particular, 
\[
\lim_{t\to+\infty} \norm{J e^{-itP_f}\f} =0.
\]
This implies $W_+\mathcal{H}_f^-=0$, and $W_-\mathcal{H}_f^+=0$ is proved 
similarly. 

Now we suppose $\mathcal{F}\f\in C_0^\infty(\re\times\pa M)$ and 
$\supp\!(\mathcal{F}\f)\subset (\d,\infty)\times \pa M$ with some $\d>0$. 
Such $\f$'s are dense in $\mathcal{H}_f^+$ and it suffices to show the claims 
(the existence of wave operators and the isometry) for $\f$. 
By the non-stationary phase method, we then have 
\[
\bignorm{\chi_{(-\infty,\d t/2)}(r)e^{-itP_f}\f} \leq C_N\jap{t}^{-N}, \quad t>0,
\]
for any $N$. This implies, in particular, 
\[
\bignorm{e^{itP} J e^{-itP_f} \f} = \bignorm{e^{itP_f}\f + (j(r)-1)e^{-itP_f}\f} 
= \norm{\f}+O(\jap{t}^{-N})
\]
as $t\to+\infty$, and hence $\norm{W_+\f}=\norm{\f}$ provided $W_+\f$ exists. 
Again by the non-stationary phase method, we learn 
\[
||| (1-j(2r/\d t))e^{-itP_f}\f ||| \leq C_N\jap{t}^{-N}, \quad t>0, 
\]
for any $N$. Thus we have 
\begin{align*}
\bignorm{T e^{-itP_f}\f} &\leq 
\bignorm{T j(2r/\d t) e^{-itP_f}\f} + \bignorm{T (1-j(2r/\d t))e^{-itP_f}\f} \\
&\leq C |||\jap{r}^{-\m} j(2r/\d t) e^{-itP_f}\f ||| + C |||(1-j(2r/\d t)) e^{-itP_f}\f ||| \\
&\leq C' \jap{t}^{-\m} +C_N' \jap{t}^{-N}, \quad \text{for }t>0
\end{align*}
by Lemma~7. This implies $\int_0^\infty \norm{T e^{-itP_f}\f}dt<\infty$, and 
combining this with (3.1), we obtain the existence of $W_+\f$. The existence 
of $W_-\f$ for $\f\in \mathcal{H}_f^-$ is similarly shown, and we omit the detail. 
\end{proof}

There are several methods available to prove the asymptotic completeness 
for one-body Schr\"odinger operators. Here we employ the stationary scattering 
theory, but the time-dependent theory (the Enss method) should work as well. 
At first, we prepare $P$-boundedness of several operators. 
We write
\[
\tilde Q=j(r)Q j(r)\quad\text{on }\mathcal{H}.
\]

\begin{lem}
(1) Let $L$ be a second order differential operator on $M$ with compactly supported 
smooth coefficients. Then $L$ is $P$-bounded. \newline
(2) $j(r)\jap{r}^{-1}\pa_\th$ is $P$-bounded. \newline
(3) $j(r)\pa_r^2$ is $P$-bounded. 
\end{lem}

\begin{proof}
(1) is a consequence of the local elliptic estimate. 
(2) follows from the fact that 
\[
|j(r)\jap{r}^{-1}\pa_\th|^2  \leq C(\tilde Q/r^2+1)
\]
with some $C>0$ in the sense of quadratic forms. 
In order to prove (3), it suffices to show
\[
j(r/2)\pa_r^4 j(r/2) \leq c_1 j(r/2)(-\pa_r^2+\tilde Q/r^2)^2j(r/2) + c_2
\]
in the quadratic form sense with some $c_1, c_2>0$. 
By direct computations, we have 
\begin{align*}
(-\pa_r^2+\tilde Q/r^2)^2 &= \pa_r^4 +\tilde Q^2/r^4 -\pa_r^2(\tilde Q/r^2) - 
(\tilde Q/r^2)\pa_r^2 \\
&= \pa_r^4 +\tilde Q^2/r^4 -2 \pa_r(\tilde Q/r^2)\pa_r -[\pa_r,[\pa_r,(\tilde Q/r^2)]].
\end{align*}
We note 
\[
[\pa_r,[\pa_r,(\tilde Q/r^2)]] =(\pa_r^2(j(r)^2/r^2))Q, 
\]
and $\pa_r^2(j(r)^2/r^2) = c r^{-4}$ for $r\geq 1$. Hence we have 
\[
j(r/2)(\tilde Q^2/r^4-[\pa_r,[\pa_r,(\tilde Q/r^2)]] )j(r/2) \geq -c' r^{-4} j(r/2)^2 \geq -c' j(r/2)^2.
\]
Since $-\pa_r(\tilde Q/r^2)\pa_r \geq 0$, we learn
\[
j(r/2)(\-\pa_r^2+\tilde Q/r^2)^2 j(r/2) \geq j(r/2) (\pa_r^4-c)j(r/2),
\]
and the assertion follows. 
\end{proof}

The next lemma is a key step in the proof of the asymptotic completeness. 
We remark $\tilde Q$ is {\it not}\/ $P$-bounded. 

\begin{lem}
Let $s>1/2$, and let $E\in (0,\infty)\setminus \s_{pp}(P)$. 
Then there exists $C>0$ such that 
\[
\bignorm{\jap{r}^{-s} \tilde Q (P-z)^{-1} (1+\tilde Q)^{-1}\jap{r}^{-s}}\leq C
\]
for $z\in \bigset{z\in\co}{\Im z\neq 0, |z-E|<\e}$. 
\end{lem}

\begin{proof}
We show 
\[
A(z) = \jap{r}^{-s} (1+\tilde Q)(P-z)^{-1} (1+\tilde Q)^{-1} \jap{r}^{-s}
\]
is bounded uniformly in a small neighborhood of $E$. 
For $R>0$, we set
\[
j_R(x)=j(r/R)\in C^\infty(M); \quad \chi_R(x)=1-j_R(x)\in C_0^\infty(M).
\]
Note, for $R>0$, we have 
\begin{align*}
\chi_R A(z) &= \jap{r}^{-s} \chi_R (1+\tilde Q)(P-i)^{-1}\jap{r}^{-s}(1+\tilde Q)^{-1} \\
&\quad + (z-i) \bigpare{\jap{r}^{-s}\chi_R (1+\tilde Q)(P-i)^{-1}\jap{r}^s} \times \\
&\quad\quad \times \bigpare{\jap{r}^{-s}(P-z)^{-1}\jap{r}^{-s}}\times (1+\tilde Q)^{-1},
\end{align*}
and it is easy to see from Lemma~9-(1) and Theorem~6 (the limiting absorption 
principle) that each term is bounded uniformly for $z$ close to $E$. 

We then consider 
\begin{align*}
j_R A(z) &= j_R \jap{r}^{-s} (1+\tilde Q)(P-z)^{-1} (1+\tilde Q)^{-1} \jap{r}^{-s}\\
&= \jap{r}^{-s} (P-z)^{-1}j_R \jap{r}^{-s} \\
&\quad + \jap{r}^{-s} (P-z)^{-1}[P, j_R(1+\tilde Q)](P-z)^{-1}(1+\tilde Q)^{-1}\jap{r}^{-s} \\
&=  \jap{r}^{-s} (P-z)^{-1}j_R \jap{r}^{-s} \\
&\quad + \jap{r}^{-s} (P-z)^{-1} [P,j_R](1+\tilde Q)(P-z)^{-1} (1+\tilde Q)^{-1}\jap{r}^{-s} \\
&\quad + \jap{r}^{-s}(P-z)^{-1}j_R [P,\tilde Q] (P-z)^{-1}(1+\tilde Q)^{-1}\jap{r}^{-s} \\
&=\mbox{I} + \mbox{II} + \mbox{III}.
\end{align*}
I is uniformly bounded by the limiting absorption principle. 
Since $[P,j_R]$ is a first order differential operator with compactly 
supported coefficients, we can easily see 
\[
\mbox{II} =\bigpare{\jap{r}^{-s}(P-z)^{-1}[P,j_R]}
\bigpare{\chi_{2R}(1+\tilde Q) (P-z)^{-1}\jap{r}^{-s}} (1+\tilde Q)^{-1}
\]
is uniformly bounded in $z\sim E$ for fixed $R>0$, as well as $\chi_R A(z)$. 
We now consider the term III. For simplicity, we write (with slight abuse of notation), 
\[
P= b_1 \pa_r^2 + b_2 \pa_r\pa_\th +b_3 \pa_\th^2 + B^{(4)}
\]
where $B^{(4)}$ is a first order differential operator with coefficients of order 
$O(\jap{r}^{-\n_4})$ where $\n_4>1$ for large $r$. 
We decompose III according to this decomposition, and estimate each term as follows: 
\begin{align*}
\mbox{III}_1 &= \jap{r}^{-s} (P-z)^{-1}j_R [b_1\pa_r^2,\tilde Q] (P-z)^{-1} 
(1+\tilde Q)^{-1}\jap{r}^{-s} \\
&= \bigbrac{\jap{r}^{-s} (P-z)^{-1}\jap{r}^{-s} (1-j\pa_r^2 j)} \times \\
&\qquad \times \bigbrac{(1-j\pa_r^2 j)^{-1}\jap{r}^s j_R [b_1\pa_r^2,\tilde Q]
(1+\tilde Q)^{-1}\jap{r}^s } \times \\
&\qquad \times \bigbrac{\jap{r}^{-s}(1+\tilde Q)(P-z)^{-1}(1+\tilde Q)^{-1} \jap{r}^{-s}}\\
&= B_R^{(1)} A(z).
\end{align*}
We recall $j\pa_r^2 j$ is $P$-bounded, and hence the first component is bounded uniformly. 
Since $[b_1\pa_r^2,\tilde Q]$ is a differential operator of the form: $\tilde b_1 \pa_r^2\pa_\th$
with $\tilde b_1=O(\jap{r}^{-\n_1})$ as $r\to+\infty$ where $\n_1>1$, we learn 
\[
\norm{B_R^{(1)}} = O(R^{-\n_1+2s}) \quad \text{as }R\to\infty.
\]
We may assume $s\in (1/2, \n_1/2)$, and then $\norm{B_R^{(1)}} = O(R^{-\d})$ 
as $R\to+\infty$ with some $\d>0$. Similarly, we have 
\begin{align*}
\mbox{III}_2 &= \jap{r}^{-s} (P-z)^{-1}j_R [b_2\pa_r\pa_\th,\tilde Q] (P-z)^{-1} 
(1+\tilde Q)^{-1}\jap{r}^{-s} \\
&= \bigbrac{\jap{r}^{-s} (P-z)^{-1}\jap{r}^{-s} (1-j\pa_r^2 j)} \times \\
&\qquad \times \bigbrac{(1-j\pa_r^2 j)^{-1}\jap{r}^s j_R [b_2\pa_r\pa_\th,\tilde Q]
(1+\tilde Q)^{-1}\jap{r}^s } \times \\
&\qquad \times \bigbrac{\jap{r}^{-s}(1+\tilde Q)(P-z)^{-1}(1+\tilde Q)^{-1} \jap{r}^{-s}}\\
&= B_R^{(2)} A(z).
\end{align*}
We note $[b_2\pa_r\pa_\th,\tilde Q]$ is a differential operator of the form: 
$\tilde b_2\pa_r\pa_\th^2$ with $\tilde b_2=O(\jap{r}^{-\n_2-1})$ as $r\to+\infty$, 
where $\n_2>0$. 
Hence, analogously to $\mbox{III}_1$, we learn  $\norm{B_R^{(2)}} = O(R^{-\d})$ 
as $R\to+\infty$. We next consider
\begin{align*}
\mbox{III}_3 & = \jap{r}^{-s} (P-z)^{-1}j_R [b_3\pa_\th^2,\tilde Q] (P-z)^{-1} 
(1+\tilde Q)^{-1}\jap{r}^{-s} \\
&= \bigbrac{\jap{r}^{-s} (P-z)^{-1}\jap{r}^{-s-1} (1+\tilde Q)^{1/2}} \times \\
&\qquad \times \bigbrac{(1-\tilde Q)^{-1/2} \jap{r}^{s+1} j_R [b_3\pa_\th^2,\tilde Q]
(1+\tilde Q)^{-1}\jap{r}^s } \times \\
&\qquad \times \bigbrac{\jap{r}^{-s}(1+\tilde Q)(P-z)^{-1}(1+\tilde Q)^{-1} \jap{r}^{-s}}\\
&= B_R^{(3)} A(z).
\end{align*}
The first component is bounded by Lemma~9-(2) and the limiting absorption 
principle. $[b_3\pa_\th^2,\tilde Q]$ is a differential operator of the form: 
$\tilde b_3\pa_\th^3$ with $\tilde b_3 =O(\jap{r}^{-\n_3-2})$ as $r\to+\infty$, 
and hence we learn $B_R^{(3)}=O(R^{-\d})$ similarly. At last, we can also show
\[
\mbox{III}_4 =\jap{r}^{-s}(P-z)^{-1} j_R B^{(4)}(P-z)^{-1}(1+\tilde Q)^{-1}\jap{r}^{-s}
\]
is uniformly bounded in $z\sim E$ for fixed $R$, similarly to $\chi_R A(z)$. 

Combining these estimates, we have
\[
A(z) =\chi_R A(z) +\mbox{II} + \mbox{III}_4 + B_R A(z), 
\]
where $B_R=B_R^{(1)} + B_R^{(2)} +B_R^{(3)}=O(R^{-\d})$, and the other terms 
are bounded uniformly in $z\sim E$ for fixed $R$. Thus, choosing $R$ sufficiently 
large, we conclude 
\[
A(z) = (1-B_R)^{-1} (\chi_R A(z) + \mbox{II}+\mbox{III}_4)
\]
is bounded uniformly for $z\sim E$. 
\end{proof}

It also follows from Lemma~9-(3) that 
\[
\bignorm{\jap{r}^{-s} j \pa_r^2 (P-z)^{-1}\jap{r}^{-s}}\leq C
\]
uniformly for $z\sim E$. Recalling Lemma~7, and using these estimates, we now have
\begin{equation}
\bignorm{\jap{r}^{s} T^* (P-z)^{-1} (1+\tilde Q)^{-1} \jap{r}^{-s}} \leq C
\quad\text{for }z\sim E
\end{equation}
if $s>1/2$ is chosen as in the proof of Lemma~10. 
The estimate (3.2) and the limiting absorption principle are sufficient to 
show the asymptotic completeness with the help of the abstract stationary 
scattering theory (see, e.g., Kato-Kuroda \cite{KK}, Yafaev \cite{Yaf1,Yaf}) in a slightly 
generalized form as described in Appendix~B. 

We set 
\begin{align*}
&\mathcal{H}_1 =\mathcal{H}_f, \quad \mathcal{H}_2 =\mathcal{H}, \quad 
\mathcal{X}_1 =\jap{r}^{-s}\mathcal{H}_f, \quad \mathcal{X}_2 =\jap{r}^{-s}\mathcal{H},\\
&\mathcal{Y}_1 = \jap{r}^{-s}(1+Q)^{-1}\mathcal{H}_f, \quad 
\mathcal{Y}_2 =\jap{r}^{-s}(1+\tilde Q)^{-1}\mathcal{H}, 
\end{align*}
where $s>1/2$ is chosen as above. We then set 
\[
H_1=P_f, \quad H_2 =P, \quad J_{12}=J^*, \quad J_{21}=J.
\]
We let $P_1^\pm$ be the orthogonal projections from $\mathcal{H}_f$ 
to $\mathcal{H}_f^\pm$, respectively, and let $P_2^\pm =1$. 
Then Assumption~C in Appendix~B is satisfied by the limiting absorption 
principle (Theorem~6 (3)) and the estimate (3.2). We recall that the limiting absorption 
principle for $P_f$ on $\mathcal{X}_1$ is well-known. 
The existence of the wave operators 
\[
W_\pm =W_{12}^\pm(\re) = \slim_{t\to\pm\infty} e^{itP} J e^{-itP_f}
\]
is proved (Theorem~8), and now we can apply Theorem~15 to conclude the asymptotic 
completeness: 

\begin{thm} 
$\Ran W_\pm=\mathcal{H}_{ac}(P) $.
\end{thm}

%%%%%%%%%%%%%%%%%%%%%%%%%%%%%%%%%%%%%%%%%
%%%%%%%%%%%%  Section 4  %%%%%%%%%%%%%%%%
%%%%%%%%%%%%%%%%%%%%%%%%%%%%%%%%%%%%%%%%%

\section{The scattering operator and the scattering matrix}

By virtue of the asymptotic completeness, the scattering operator:
$S= W_+^* W_-$ : $\mathcal{H}_f^- \to \mathcal{H}_f^+$ 
is a unitary operator. 
For $\f\in \mathcal{H}_f^\pm$, we set 
\[
(F_{0,\pm}(\l)\f)(\th) =(2\l)^{-1/4} (\mathcal{F}\f)(\pm\sqrt{2\l},\th), \quad\l>0, \th\in\pa M.
\]
If $\f\in \jap{r}^{-s}\mathcal{H}_f^\pm$ with some $s>1/2$, then 
$F_{0,\pm}(\l)\f \in \mathcal{H}_b=L^2(\pa M,Hd\th)$, and $F_{0,\pm}(\cdot)$ is 
extended to a unitary operator from $\mathcal{H}_f^\pm$ to $L^2(\re_+;\mathcal{H}_b)$. 
In fact, by the Plancherel theorem, we have 
\begin{align*}
\int_0^\infty \bignorm{ F_{0,\pm}(\l)\f}_{\mathcal{H}_b}^2 d\l 
&= \int_0^\infty \int_{\pa M} \bigabs{(\mathcal{F}\f)(\pm\sqrt{2\l},\th)}^2 H(\th) d\th \frac{d\l}{\sqrt{2\l}} \\
&= \int_0^\infty \int_{\pa M} \bigabs{(\mathcal{F}\f)(\pm\x,\th)}^2 H(\th)d\th d\x
=\norm{\f}^2_{\mathcal{H}_f^\pm}
\end{align*}
for $\f\in \mathcal{H}_f^\pm$. It is also easy to see 
\[
F_{0,\pm}(\l) P_f \f = \l F_{0,\pm}(\l)\f \quad\text{for }\f\in \mathcal{D}(P_f),
\]
and hence $F_{0,\pm}(\cdot)$ is a spectral representation of $P_f$ on $\mathcal{H}_f^\pm$. 
We note $S$ satisfies the intertwining property, i.e., for any $g\in L^\infty(\re)$, 
\[
g(P_f)\, S = S \,g(P_f) \ : \ \mathcal{H}_f^- \to \mathcal{H}_f^+.
\]
Then by the abstract theory of scattering, we learn there exists 
$S(\l)\in \mathcal{L}(\mathcal{H}_b)$ for a.e.\ $\l>0$ such that 
\[
F_{0,+}(\l) S\f = S(\l) F_{0,-}(\l)\f
\quad\text{for }\f\in \jap{r}^{-s}\mathcal{H}_f^-.
\]
$S(\l)$ is the scattering matrix for our scattering system $(P,P_f;J)$. 
Moreover, by the abstract scattering theory, we have the following representation 
formula for $S(\l)$ for $\l\in\re_+\setminus \s_{pp}(P)$. (Appendix~B. 
See also \cite{Yaf} Section~7.3.)
\begin{equation}
S(\l) =-2\pi i F_{0,+}(\l)\bigpare{J^* T -T^* (P-\l-i0)^{-1}T} F_{0,-}(\l)^*.
\end{equation}

On the other hand, we can represent the wave operators as follows (Appendix B):  
Let $I\subset\subset \re_+\setminus \s_{pp}(P)$, then 
\begin{equation}
W_\pm E_{P_f}(I) \f =\int_I \bigpare{J-(P-\l\pm i0)^{-1} T}F_{0,\pm}(\l)^*F_{0,\pm}(\l)\f d\l.
\end{equation}
In particular, for $\f\in \mathcal{D}(Q)$,
\[
\G = \bigpare{J-(P-\l-i0)^{-1} T}F_{0,-}(\l)^*\f \in \jap{r}^s\mathcal{H}, \quad  s>1/2, 
\]
is a generalized eigenfunction of $P$ with the eigenvalue $\l$. We now consider 
the asymptotic behavior of $\G$ as $r\to+\infty$. 
By the resolvent equation, we have 
\begin{align*}
J^* \G &= \bigpare{J^* J - J^*(P-\l-i0)^{-1} T} F_{0,-}(\l)^* \f \\
&= J^*J F_{0,-}(\l)^* \f \\
&\quad - (P_f-\l-i0)^{-1} \bigpare{J^*T- T^*(P-\l-i0)^{-1} T}F_{0,-}(\l)^*\f.
\end{align*}
The first term is easy to handle: 
\[
J^*J F_{0,-}(\l)^*\f = \frac{1}{\sqrt{2\pi k}} e^{-ik r}\f(\th) 
\quad \text{for } r\geq 1, 
\]
where $k=\sqrt{2\l}$. 

\begin{lem}
$(-\frac12 \pa_r^2-z)^{-1}$, $z\notin \re_+$, has the following integral kernel: 
\[
(-\tfrac{1}{2}\pa_r^2 -z)^{-1}\f(r) = \int_{-\infty}^\infty (i/\z) e^{i|r-r'|\z} \f(r')dr',
\]
where $z=\z^2/2$ with $\Im \z>0$. 
\end{lem}

The proof of this lemma is an elementary computation of the Fourier analysis. 

We denote 
$\F = \bigpare{J^* T -T^*(P-\l-i0)^{-1} T} F_{0,-}(\l)^*\f$. Then we have 
\begin{align*}
& J^*\G-J^*JF_{0,-}(\l)^*\f = -\frac{i}{k}\int_{-\infty}^\infty e^{i|r-r'|k} \F(r',\cdot)dr' \\
&\quad =-\frac{i}{k}\biggpare{\int_{-\infty}^r e^{i(r-r')k} \F(r',\cdot)dr' 
+\int_r^\infty e^{-i(r-r')k} \F(r',\cdot)dr'}\\
&\quad=  \G_1 +\G_2.
\end{align*}
Since $\F\in \jap{r}^{-s} \mathcal{H}_f$ with $s>1/2$, we learn 
$\norm{\G_2(r,\cdot)}_{\mathcal{H}_b} =O(\jap{r}^{-\d})$ as $r\to+\infty$ 
with some $\d>0$. Similarly, if we set
\[
\G_3=-\frac{i}{k} e^{irk} \int_{-\infty}^\infty e^{-ir'k}\F(r',\cdot) dr',
\]
then $\norm{\G_1(r,\cdot)-\G_3(r,\cdot)}_{\mathcal{H}_b} = O(\jap{r}^{-\d})$ 
as $r\to+\infty$.  Moreover, if we compare $\G$ with $S(\l)\f$ using (4.1), 
we learn
\[
\G_3 =e^{irk} \frac{i}{k} \sqrt{2\pi k} F_{0,+}(\l)\F =  \frac{1}{\sqrt{2\pi k}} e^{irk} (S(\l)\f).
\]
Combining these, we have
\[
\biggnorm{ J^*\G(r,\cdot) -\frac{1}{\sqrt{2\pi k}} \bigpare{e^{-ikr}\f + 
e^{ikr} (S(\l)\f)}}_{\mathcal{H}_b}
=O(\jap{r}^{-\d})
\]
as $r\to+\infty$ with some $\d>0$. Replacing $\f$ by $\sqrt{2\pi k}\, \f$, and applying 
$J$ to these terms, we conclude the following:

\begin{thm}
Let $\f\in \mathcal{D}(Q)$ and let $\l\in\re_+\setminus\s_{pp}(P)$. Then there exists 
$\G$, a generalized eigenfunction of $P$ with the eigenvalue $\l$ such that 
$\G\in \jap{r}^{s}\mathcal{H}$ for $s>1/2$ and 
\[
\biggnorm{ \G(r,\cdot) - r^{-(n-1)/2} \bigpare{e^{-ikr}\f + e^{ikr} (S(\l)\f)}}_{\mathcal{H}_b}
=O(\jap{r}^{-(n-1)/2-\d}) 
\]
as $r\to+\infty$ with some $\d>0$, where $k=\sqrt{2\l}$. 
\end{thm}

The eigenfunction $\G$ in the theorem should be unique, though we do not 
prove it here (cf.~\cite{MZ}). This result implies that our scattering matrix $S(\l)$ is actually 
the same as the {\it absolute scattering matrix} of Melrose. 
In other words, the absolute scattering matrix can be defined in terms of 
the traditional time-dependent scattering theory. 

%%%%%%%%%%%%%%%%%%%%%%%%%%%%%%%%%%%%%%%%%
%%%%%%%%%%%%  Appendix  A %%%%%%%%%%%%%%%%
%%%%%%%%%%%%%%%%%%%%%%%%%%%%%%%%%%%%%%%%%

\appendix
\section{Geometric Construction}

In this appendix, we discuss relationship between our model and the Laplace-Beltrami 
operator on Riemannian manifolds, in particular, manifolds with scattering metrics. 

Let $M$ be a manifold with a smooth compact boundary $\pa M$. Let $x$ be a boundary defining 
function, i.e., $x$ is a smooth function on $\overline M$ such that $\pa M=\{p\in \overline{M}
\,|\, x(p)=0\}$. A metric $g$ on $M$ is called a {\it scattering metric in the sense of Melrose},
if 
\[
g=\frac{dx^2}{x^4} +\frac{\tilde h}{x^2}
\]
in a neighborhood of $\pa M$, where $\tilde h$ is a $(0,2)$-tensor such that 
$\tilde h$ is smooth up to the boundary, and the pull-back of $\tilde h$ to $\pa M$ 
defines a Riemannian metric on $\pa M$. 

We introduce the {\it polar coordinate}: 
\[
(r,\th)=(x^{-1},\th)\in (\c,\infty)\times \pa M
\]
in a neighborhood of $\pa M$ with some $\c>0$. 
In this coordinate, the above metric $g$ is expressed as 
\[
g= dr^2+r^2\tilde h(r,\th,-dr/r^2,d\th). 
\]
By the assumption, $\tilde h$ has an asymptotic expansion in $r$ as $r\to\infty$:
\begin{align*}
\tilde h &\sim \sum_{j=0}^\infty h_j(\th,dr/r^2,d\th) r^{-j} \\
&\sim \sum_{j=0}^\infty h_j^1(\th) dr^2\, r^{-4-j} + 
\sum_{j=0}^\infty h_j^2(\th) dr\,d\th\, r^{-2-j} +
\sum_{j=0}^\infty h_j^3(\th) d\th^2\, r^{-j}
\end{align*}
as $r\to\infty$. Hence we can write
\begin{equation}
g= m_1(r,\th) dr^2 + m_2(r,\th)dr(rd\th) +m_3(r,\th) (rd\th)^2, 
\end{equation}
where $m_j$, $j=1,2,3$, satisfy 
\begin{align*}
&\bigabs{\pa_r^\ell\pa_\th^\a (m_1(r,\th)-1)} \leq C_{\ell\a}\jap{r}^{-2-\ell},\\
&\bigabs{\pa_r^\ell\pa_\th^\a m_2(r,\th)} \leq C_{\ell\a}\jap{r}^{-1-\ell}, \\
&\bigabs{\pa_r^\ell\pa_\th^\a (m_3(r,\th)-h_0^3(\th))} \leq C_{\ell\a}\jap{r}^{-1-\ell} 
\end{align*}
for any $\ell$, $\a$ and $(r,\th)\in (\c,\infty)\times\pa M$. 
By the assumption, $h_0^3$ is a positive tensor on $\pa M$. 
More generally, we call $g$ a {\it scattering metric}\/ if $g$ has the above expression 
(A.1) and $m_j$ satisfy
\begin{align*}
&\bigabs{\pa_r^\ell\pa_\th^\a (m_1(r,\th)-1)} \leq C_{\ell\a}\jap{r}^{-\n_1-\ell},\\
&\bigabs{\pa_r^\ell\pa_\th^\a m_2(r,\th)} \leq C_{\ell\a}\jap{r}^{-\n_2-\ell}, \\
&\bigabs{\pa_r^\ell\pa_\th^\a (m_3(r,\th)-h_0^3(\th))} \leq C_{\ell\a}\jap{r}^{-\n_3-\ell} 
\end{align*}
for any $\ell$, $\a$ and $(r,\th)\in (\c,\infty)\times\pa M$, where $\n_1>1$, $\n_2>1/2$
and $\n_3>0$. We denote $g^*=(g^{jk})= (g_{jk})^{-1}$ be the Riemaniann metric on 
$T^*M$. If $g$ is a scattering metric, then $g^*$ is also expressed as follows: 
\[
g^*= b_1(r,\th)\, \pa_r^2 + b_2(r,\th)\pa_r\,(r^{-1}\pa_\th)+b_3(r,\th)\, (r^{-1}\pa_\th)^2, 
\]
where $b_j$, $j=1,2,3$, satisfy 
\begin{align*}
&\bigabs{\pa_r^\ell\pa_\th^\a (b_1(r,\th)-1)} \leq C_{\ell\a}\jap{r}^{-\n_1'-\ell},\\
&\bigabs{\pa_r^\ell\pa_\th^\a b_2(r,\th)} \leq C_{\ell\a}\jap{r}^{-\n_2'-\ell}, \\
&\bigabs{\pa_r^\ell\pa_\th^\a (b_3(r,\th)-h(\th))} \leq C_{\ell\a}\jap{r}^{-\n_3'-\ell} 
\end{align*}
for any $\ell$, $\a$ and $(r,\th)\in (\c,\infty)\times\pa M$, where
$\n_1'>1$, $\n_2'>1/2$, $\n_3'>0$, and $h=(h_0^3)^{-1}$ is a positive metric tensor on 
$T^*(\pa M)$. (Note that here we use $(\pa_r,\pa_\th)$ to represent the basis of the 
tangent space $T_x M$, $x=(r,\th)\in M$.) 

We now consider the Laplace-Beltrami operator $\triangle_g$ on $C^\infty(M)$:
\[
\triangle_g = \frac{1}{\sqrt{\det g}} \sum_{j,k} \pa_{x_j}g^{jk}\sqrt{\det g}\  \pa_{x_k}.
\]
It is well-known that $-\triangle_g$ is a non-negative and symmetric operator 
on $\mathcal{H}'=L^2(M,\sqrt{\det g}\, dx)$. 
The principal symbol of $\triangle_g$ is given by $g^{jk}(x)\x_j\x_k$, and they satisfy the 
above condition, which is almost identical to Assumption~A. However, we use different 
measure to define the Hilbert spaces $\mathcal{H}$ and $\mathcal{H}'$, 
though $\sqrt{\det g}$ asymptotically converges to $G$ as $r\to\infty$. 
We construct an identification operator $U: \mathcal{H}' \to \mathcal{H}$ as follows: 
Let
\[
\G(x)=\biggpare{\frac{\sqrt{\det g(x)}}{G(x)}}^{1/2} \quad \text{for }x\in M, 
\]
and set a unitary operator $U$ from $\mathcal{H}'$ to $\mathcal{H}$ by 
\[
(U\f)(x) =\G(x)\f(x) \quad \text{for }\f\in \mathcal{H}'.
\]
It is easy to show 
\[
\bigabs{\pa_r^\ell\pa_\th^\a(\G(r,\th)-1)}\leq C_{\ell\a}\jap{r}^{-\n-\ell}
\quad\text{for } (r,\th)\in (\c,\infty)\times\pa M
\]
for any $\ell$, $\a$ with some $\n>0$. Then we set
\[
P=-\frac12 U\triangle_g U^{-1} \quad\text{on }\mathcal{H}=L^2(M,Gdx).
\]
By straightforward computations, we have 
\[
U\triangle_g U^{-1} = G^{-1} \sum_{j,k} \pa_{x_j} g^{jk} G \pa_{x_k} +W,
\]
where 
\[
W= \sum_{j,k} g^{jk}\G^{-2}(\pa_{x_j}\G)(\pa_{x_k}\G) 
+ \sum_{j,k} G^{-1} \pa_{x_k}\bigpare{g^{jk}G\G^{-1}(\pa_{x_j}\G)}.
\]
We can easily show $W$ satisfies the condition for $V$ in Assumption~A with 
$\m_4>2$. Thus $P$ satisfies Assumption~A, and since $P$ is unitarily equivalent to 
$-\frac12\triangle_g$, our results of this paper apply to perturbations of the 
Laplace-Beltrami operator as well. 

%%%%%%%%%%%%%%%%%%%%%%%%%%%%%%%%%%%%%%%%%
%%%%%%%%%%%%  Appendix B   %%%%%%%%%%%%%%%%
%%%%%%%%%%%%%%%%%%%%%%%%%%%%%%%%%%%%%%%%%

\section{Abstract stationary scattering theory}

Here we give a formulation of the abstract stationary scattering theory 
suitable for our application in Section~3 and Section~4. 
We do not give full proof of theorems below, but only make remarks 
on the necessary modifications. See, e.g., Kato-Kuroda \cite{KK}
or Yafaev \cite{Yaf, Yaf1} for the full discussion of the abstract stationary 
scattering theory. 

Let $\mathcal{H}_j$, $j=1,2$, be Hilbert spaces, and let $H_j$ be self-adjoint 
operators on $\mathcal{H}_j$, $j=1,2$, respectively. Let $\mathcal{X}_j$, 
$\mathcal{Y}_j$ be Banach spaces such that 
\[
\mathcal{Y}_j \subset \mathcal{X}_j \subset \mathcal{H}_j \subset 
\mathcal{X}_j^* , \quad j=1,2,
\]
and $\mathcal{Y}_j$ are dense in $\mathcal{H}_j$. We suppose 
$(H_j+i)^{-1} \mathcal{Y}_j\subset \mathcal{Y}_j$. 
The embedding $\mathcal{H}_j\subset \mathcal{X}_j^*$  are defined 
through the standard pairing in $\mathcal{H}_j$. 
Let 
\[
J_{21}\in \mathcal{L}(\mathcal{H}_1,\mathcal{H}_2), \quad J_{21}\in
\mathcal{L}(\mathcal{H}_2,\mathcal{H}_1)
\]
be bounded operators such that $J_{12}=(J_{21})^*$, and we suppose 
\[ 
J_{jk}(\mathcal{Y}_k\cap\mathcal{D}(H_k)) \subset  \mathcal{D}(H_j) \quad 
\text{for } (j,k)=(1,2) \text{ or } (2,1).
\]
We fix a bounded interval $I\subset\subset \re$, and consider the scattering 
theory on the energy interval $I$. We write 
\[
R_j(z)=(H_j-z)^{-1}
\]
for $j=1,2$, $z\in\co\setminus\re$, and 
\[
G_{jk}(z) =(H_j-z)J_{jk}R_k(z) = J_{jk} +T_{jk} R_k(z)\ :\ 
\mathcal{H}_j \to \mathcal{H}_k 
\]
for  $(j,k)=(1,2)$ or $(2,1)$, where 
\[
T_{jk} = H_j J_{jk}- J_{jk} H_k. 
\]
We note $\mathcal{D}(G_{jk}(z)) \supset \mathcal{Y}_k$, and 
$\mathcal{D}(T_{jk})\supset \mathcal{D}(H_k)\cap \mathcal{Y}_k$. 

\begin{ass} (1) Let $\l\in I$ and $j=1,2$. Then 
\[
R_j(\l\pm i0) =\slim_{\e\downarrow 0} R_j(\l\pm i\e)\in 
\mathcal{L}(\mathcal{X}_j, \mathcal{X}_j^*)
\]
exist. Moreover, for $\f\in\mathcal{X}_j$, $\l\mapsto R_j(\l\pm i0)\f$ are 
strongly continuous in $\l\in I$. \newline
(2) Let $(j,k)=(1,2)$ or $(2,1)$. For $z\in\co\setminus\re$, $G_{jk}(z)\in 
\mathcal{L}(\mathcal{Y}_k,\mathcal{X}_j)$, and for any $\l\in I$, 
\[
G_{jk}(\l\pm i0) =\slim_{\e\downarrow 0} G_{jk}(\l\pm i\e) 
\in \mathcal{L}(\mathcal{Y}_k, \mathcal{X}_j)
\]
exist. Moreover, for $\f\in\mathcal{Y}_k$, $G_{jk}(\l\pm i0)\f$ are strongly continuous 
in $\l\in I$. \newline
(3) There exist orthogonal projections $P_1^\pm$, $P_2^\pm$ on 
$\mathcal{H}_1$, $\mathcal{H}_2$, respectively, such that 
\[
\lim_{t\to \pm \infty} \bignorm{J_{jk} e^{-it H_k}\f} =\bignorm{P_k^\pm\f}
\quad \text{for }\f\in \mathcal{H}_k.
\]
Moreover, $P_j^\pm$ commute with $H_j$, $j=1,2$. 
\end{ass}

We denote $E_j(\cdot) =E_{H_j}(\cdot)$ for $j=1,2$. 
The condition (1) implies 
\[
E'_j(\l) =\frac{1}{2\pi i} \bigpare{R_j(\l+i 0) - R_j(\l -i 0)}
\in \mathcal{L}(\mathcal{X}_j,\mathcal{X}_j^*)
\]
is well-defined. Let $I'\subset\subset I$, and we define
\[
\Omega_{jk}^\pm(I') \, \f = \int_{I'} E_j'(\l) G_{jk}(\l\pm i0)\,\f \, d\l \ \in\  \mathcal{X}_j^*
\quad \text{for $\f\in\mathcal{Y}_k$}. 
\]

\begin{thm} Let $I'\subset\subset I$ as above, and let $(j,k)=(1,2)$ or $(2,1)$. 
Then: \newline
(1) $\Omega_{jk}^\pm(I')\f\in \mathcal{H}_j$ for $\f\in\mathcal{Y}_k$. Moreover,  
$\Omega_{jk}^\pm(I')$ is extended to a partial isometry from $\mathcal{H}_k$ to 
$\mathcal{H}_j$ with the initial subspace $P_k^\pm E_k(I')\mathcal{H}_k$ 
and the range $P_j^\pm E_j(I') \mathcal{H}_j$. \newline
(2) $\Omega_{jk}^\pm(I')^* =\Omega_{kj}^\pm(I')$. \newline
(3) $\Omega_{jk}^\pm(I')$ satisfies the intertwining property, i.e., 
for any $f\in L^\infty(\re)$, 
\[
f(H_j) \Omega_{jk}^\pm(I') =\Omega_{jk}^\pm(I') f(H_k). 
\]
\end{thm}

We make several remarks on the proof. At first, the introduction of 
another auxiliary spaces $\mathcal{Y}_j$ does not change the argument, 
and we can just follow the standard procedure. 

Also by the standard argument, we can easily observe 
$\Omega_{jk}^\pm(I')E_k(F)=0$ if $I'\cap F=\emptyset$, and hence 
\[
\Omega_{jk}^\pm(I') = \Omega_{jk}^\pm(I') E_k(I').
\]
On the other hand, by Assumption~C-(3), we have 
\[
\lim_{\e\downarrow 0} 2\e\int_{\re_\pm} \bignorm{J_{jk} e^{-it(H_k-i\e)}
\f}^2 dt =\bignorm{P_k^\pm \f}^2.
\]
Then by the Plancherel theorem, we learn 
\[
\lim_{\e\downarrow 0} \frac{\e}{\pi} \int_{-\infty}^\infty \bignorm{J_{jk} R_k(\l\pm i\e)\f}^2
 d\l  =\bignorm{P_k^\pm \f}^2. 
\]
For $\f\in \mathcal{Y}_k$, we also have 
\begin{align*}
\bignorm{\Omega_{jk}^\pm(I')\f}^2 
&= \lim_{\e\downarrow 0} \frac{\e}{\pi} \int_{I'} \bignorm{J_{jk} 
R_k(\l\pm i\e)E_k(I')\f}^2 d\l  \\
&= \lim_{\e\downarrow 0} \frac{\e}{\pi} \int_{-\infty}^\infty \bignorm{J_{jk} 
R_k(\l\pm i\e) E_k(I')\f}^2 d\l
=\bignorm{P_k^\pm E_k(I') \f}^2
\end{align*}
and this implies $\Omega_{jk}^\pm(I')$ is a partial isometry from 
$P_k^\pm E_k(I')$ into $\mathcal{H}_j$. 

The rest of the proof works as in the one-space case. 
If we know the existence of the wave operators, $\Omega_{jk}^\pm(I')$ 
are in fact the wave operators, and we can conclude the asymptotic completeness. 

\begin{thm}
Suppose Assumption~C, and let either $(j,k)=(1,2)$ or $(2,1)$. Suppose 
\[
\int_{\re_\pm} \bignorm{T_{jk} e^{-itH_k} \f} dt < \infty
\]
for $\f\in \mathcal{D}$, a dense subspace of $\mathcal{Y}_k$. Then 
\[
W^\pm_{jk}(I') =\slim_{t\to\pm\infty} e^{itH_j} J_{jk} e^{-itH_k} E_k(I')
\]
exists and $W_{jk}^\pm(I')= \Omega_{jk}^\pm(I')$. Moreover, if 
$P_j^\pm=1$, then $W_{jk}^\pm(I')$ are unitary operators from 
$P_k^\pm E_k(I') \mathcal{H}_k$ to $E_j(I')\mathcal{H}_j$. 
In other words, $W_{jk}^\pm(I')$ are asymptotically complete on $I'$. 
\end{thm}

Finally, we give a proof of the representation formulas (4.1) and (4.2). 
The following computation is slightly formal, but the justification 
is straightforward. Let $F_1(\cdot)$ be a spectral representation 
of $H_1$ as used in Section~4. Then by the Cook-Kuroda method, 
for $\f\in\mathcal{Y}_1$ we have 
\begin{align*}
W_{21}^\pm E'_1(\l)\f 
&= \biggpare{J_{21} +\lim_{\e\downarrow 0} \int_0^{\pm\infty} 
e^{itH_2}(iT_{21}) e^{-itH_1} e^{-\e|t|} dt} E_1'(\l)\f  \\
&= \biggpare{J_{21} +\lim_{\e\downarrow 0} i \int_0^{\pm\infty}
e^{it(H_2-\l\pm i\e)} T_{21} dt} E_1'(\l) \f \\
&= \bigpare{J_{21} -R_2(\l\mp i0)T_{21}} E_1'(\l)\f. 
\end{align*}

Noting $E_1'(\l) =F_1(\l)^* F_1(\l)$, we learn 
\begin{equation}
W_{21}^\pm E_1(I')\f = 
\int_{I'} \bigpare{J_{21} -R_2(\l\mp i0)T_{21}} F_1(\l)^* F_1(\l)\f\, d\l, 
\end{equation}
and (4.2) follows. 

Concerning the scattering matrix, we write $W_\pm=W_{21}^\pm$ and 
$S=W_+^* W_-$. Then we have 
\[
S= W_+^* W_+ - W_+^*(W_+-W_-).
\]
By Assumption~C-(3), we can easily see $W_+^* W_+ = P_1^+$. 
On the other hand, $S$ is a map from $\Ran P_1^-$ to $\Ran P_1^+$, and 
hence we may write the first term as $P_1^+ P_1^-$. 

The second term is computed as follows: by (B.1) we have 
\begin{align*}
(W_+-W_-)E_1'(\l) &= \bigpare{-R_2(\l-i0)T_{21} 
+R_2(\l+i0) T_{21}} E_1'(\l) \\
&= 2\pi i E_2'(\l) T_{21} E_1'(\l), 
\end{align*}
and we note $W_+ E_1'(\l) = E_2'(\l) W_+$. 
Using these, for $\f,\g\in\mathcal{Y}_1$ we have 
\begin{align*}
&\bigpare{\f, W_+^* (W_+-W_-)\g} 
=  \int_{I'}\int_{I'} \bigpare{W_+ E_1'(\l)\f, 2\pi i E_2'(\l') T_{21} E_1'(\l') \g}d\l\,d\l' \\
&\qquad = 2\pi i \int_{I'} \int_{I'} \bigpare{W_+\f, E_2'(\l)E_2'(\l') T_{21} E_1'(\l')\g} d\l\, d\l' \\
&\qquad = 2\pi i \int_{I'} \bigpare{W_+ E_1'(\l) \f, T_{21} E_1'(\l)\g} d\l \\
&\qquad = 2\pi i\int_{I'} \bigpare{(J_{21}-R_2(\l-i0)T_{21})E_1'(\l)\f, T_{21}E_1'(\l)\g} d\l \\
&\qquad = 2\pi i \int_{I'} \Bigpare{ F_1(\l)\f, F_1(\l)\bigpare{J_{21}^* T_{21} 
-T_{21}^* R_2(\l+i0) T_{21}} F_1(\l)^* F_1(\l)\g} d\l.
\end{align*}
These imply 
\[
\bigpare{\f,(S-P_1^+P_1^-)\g} = \int_{I'} \bigpare{F_1(\l)\f,(-2\pi i)T(\l) F_1(\l)\g}d\l,
\]
where $T(\l)$ is the T-matrix defined by
\[
T(\l) = F_1(\l) \bigpare{J_{21}^* T_{21} - T_{21}^* R_2(\l+i0) T_{21}} F_1(\l)^*.
\]
In the one-space scattering theory, we choose $P_1^\pm=1$, and hence 
$S(\l)=1-2\pi i T(\l)$ is the scattering matrix. However, in our application in 
Section~4, $P_f^+ P_f^- =0$, hence $S(\l)=-2\pi i T(\l)$, and we obtain the 
formula (4.1).  \qed

%%%%%%% Bibliography %%%%%%%%%

\end{document}